\documentclass[runningheads]{llncs}
\usepackage{graphicx}
\usepackage{amsmath}
\usepackage{amssymb}
\usepackage{bm}
\usepackage{dsfont}
\usepackage{graphicx}
\usepackage{subfigure}
\usepackage{dsfont}
\usepackage{algorithm}
\usepackage{algorithmic}
\usepackage{float}

\DeclareMathOperator*{\argmax}{arg\,max}

\newcommand{\Real}{\mathbb{R}}

\newcommand{\veca}{\bm{a}}
\newcommand{\vecb}{\bm{b}}
\newcommand{\vecx}{\bm{x}}
\newcommand{\vecy}{\bm{y}}
\newcommand{\vecz}{\bm{z}}
\newcommand{\vecw}{\bm{w}}
\newcommand{\vecl}{{\bm{l}}}
\newcommand{\vecr}{\bm{r}}
\newcommand{\vecs}{\bm{s}}
\newcommand{\vecu}{{\bm{u}}}
\newcommand{\vecg}{\bm{g}}
\newcommand{\vecalpha}{\bm{\alpha}}

\newcommand{\norm}[1]{\left\lVert #1 \right\rVert}
\newcommand{\ip}[2]{\langle #1, #2 \rangle}

\newtheorem{assumption}[theorem]{Assumption}

\begin{document}
\title{Improved algorithms for online load balancing\thanks{Supported by organization x.}}

\author{
Yaxiong Liu\inst{1,3}
 \and Kohei Hatano\inst{2,3} \and  Eiji Takimoto\inst{1} \\
 \institute{Department of Informatics, Kyushu University, Japan\and Faculty of Arts and Science, Kyushu University, Japan \and RIKEN AIP, Japan}
 \email{\{yaxiong.liu,hatano,eiji\}@inf.kyushu-u.ac.jp\inst{1}\inst{2}\inst{3} }
}

\authorrunning{Yaxiong Liu \and Kohei Hatano \and Eiji Takimoto}

\iffalse
\author{Yaxiong Liu\inst{1}\orcidID{0000-0001-8588-4543} \and
Kohei Hatano\inst{2}\orcidID{0000-0002-1536-1269} \and
Eiji Takimoto\inst{3}\orcidID{}}
%
\authorrunning{F. Author et al.}
\institute{Information Department Kyushu University \email{yaxiong.liu@inf.kyushu-u.ac.jp}\and
Faculty of Arts and Science, Kyushu University/ RIKEN, Japan
\email{hatano@inf.kyushu-u.ac.jp} \and
Department of Informatics, Kyushu University, Japan
\email{eiji@inf.kyushu-u.ac.jp}}
\fi

\maketitle              % typeset the header of the contribution
\begin{abstract}
We consider an online load balancing problem and its extensions in the
framework of repeated games.
On each round, the player chooses a distribution (task allocation)
over $K$ servers, and then
the environment reveals the load of each server, which determines
the computation time of each server for processing the task assigned.
After all rounds, the cost of the player is measured by
some norm of the cumulative computation-time vector.
The cost is the makespan if the norm is $L_\infty$-norm.
The goal is to minimize the regret, i.e., minimizing
the player's cost relative to
the cost of the best fixed distribution in hindsight.
We propose algorithms for general norms and prove their regret bounds.
In particular, for $L_\infty$-norm, our regret bound matches the best
known bound and the proposed algorithm runs in
polynomial time per trial involving linear programming and second order
programming, whereas no polynomial time algorithm was previously known
to achieve the bound.

\end{abstract}
\section{Introduction}
We consider an online load balancing problem defined as follows.
%the online learning problem for global cost functions posed
%by Even-Dar et al.~\cite{even2009online}.
%A motivating example of the problem is an online load balancing
%problem defined as follows.
%The online load balancing problem has been extensively studied in the
%literature of online algorithms.
%Online load balancing problem plays an important role in application of
%machine learning in reality.
%This problem is defined as follow:
There are $K$ parallel servers and the protocol is defined as a game
between the player and the environment.
On each round $t=1,\dots, T$,
(i) the player selects a distribution $\vecalpha_{t}$ over $K$ servers, which can be
viewed as an allocation of data, (ii) then the environment assigns
a loaded condition $l_{t,i}$ for each server $i$ and the loss of  server $i$ is given
as $\alpha_{t,i}l_{t,i}$.
The goal of the player is to minimize the makespan of the cumulative loss vector of
all servers after $T$ rounds, i.e., $\max_{i=1,\dots,K} \sum_{t=1}^T
\alpha_{t,i}l_{t,i}$, compared relatively to the makespan obtained by
the optimal static allocation $\vecalpha^*$ in hindsight.
More precisely, the goal is to minimize the regret, the difference
between the player's makespan and the static optimal makespan.
The makespan cost can be viewed as  $L_\infty$-norm of the vector of
cumulative loss of each server (we will give a formal definition of the problem in the next section).

In traditional literature the measurement of an algorithm is always competitive ratio(e.g.,\cite{azar:bookchap98} \cite{molinaro2017online}).
In our paper we utilize another well-known measurement as ``Regret'' defined in later section.
Even-Dar et al.\cite{even2009online} gave an algorithm based on the regret
minimum framework by involving an extra concept, the Blackwell
approachability \cite{blackwell1956analog} with respect to $L_{2}$-norm,
to a target set, which is defined in the following section. This algorithm achieves the regret bound as $O(\sqrt{KT}).$
Simultaneously another algorithm, DIFF, achieves the regret upper bound as $O((\ln K)\sqrt{T}).$
Rahklin et al. \cite{rakhlin2011online} gave a theoretical result for
the online load balancing problem, that the upper bound to regret can
achieve $O(\sqrt{(\ln K)T}),$ rather than $O((\ln K)\sqrt{T}).$
However there is no efficient algorithm given in this paper to obtain
the regret.

In following years,
there were some new explorations about the equivalence
between the Blackwell approachability and online linear optimization(OLO)
\cite{abernethy2011blackwell},
in addition and online convex optimization(OCO) by involving a support
function \cite{shimkin2016online}.

These work \cite{abernethy2011blackwell} \cite{shimkin2016online} implied that the Blackwell approachability can be given by general norm by reducing Blackwell approaching game
to an OCO problem. Moreover due to this result we give an efficient algorithm to online load balancing problem, achieving the best known regret.

More specifically speaking, we propose algorithms for online load balancing for
arbitrary norms under a natural assumption.
And our technical contributions are the following:
%In conclusion we can construct an algorithm for online global cost function problem with a potential better upper bound to regret, based on a faster
%%convergence speed to the target set. This faster convergence rate is
%guaranteed by an online convex optimization algorithm. Our main
%contributions in this paper are in the following:
\begin{itemize}
\item 1. We propose a new reduction technique from online load balancing
      to a Blackwell approaching game. This reduction enables us to use
      more general norms than $L_2$-norm or $L_\infty$-norm used in the
      previous work.
      Then, by using the reduction technique of
      Shimkin~\cite{shimkin2016online} from Blackwell games to online
      linear optimization, we reduce online load balancing to online
      linear optimization.
      %a reduction from online learning with global cost
      %function with respect to general norms to online convex
      %optimization.
      %This reduction is composed by two sub-reductions.
      %The first one is a reduction from online learning with global cost
      %function with respect to some norm to a Blackwell approaching game
      %associated with the norm.
      %Then, the second reduction is from the Blackwell approaching game
      %to online convex optimization .
      %However the second reduction is not a trivial utilization of \cite{shimkin2016online}, since we can not guarantee that the target set is approachable in advance.

\item 2. Especially we give an efficient algorithm for online load
      balancing w.r.t.~$L_\infty$-norm, achieving the best known $O(\sqrt{T\ln K})$
      regret. The algorithm involves linear programming and the second
      order cone programming and runs in polynomial time per trial. This
      is the first polynomial time algorithm achieving $O(\sqrt{T\ln
      K})$ regret.
      %On each round utilizing Linear Programming to calculate the allocation distribution $\bm{\alpha}_{t}$ and Second Order Cone Programming(SOCP) \cite{lobo1998applications} to sub-gradient for online convex optimization algorithm.

\end{itemize}

This paper is organized as follows. In section 2 we introduce the basic
definitions in this paper like online load balancing problem, Blackwell
approachability game and online convex optimization. Next in section 3
we give a meta algorithm
for online load balancing with respect to any norm under a natural assumption. Then in section 4 we give some details in
implementation of the algorithm for $L_\infty$-norm.

\section{Preliminaries}
First we give some notations.
We use $\norm{\cdot}$ to denote a norm of a vector.
More specifically,
for a vector $\vecx = (x_1, x_2, \ldots, x_d) \in \Real^d$
and a real number $p \geq 1$, the $L_p$-norm of $\vecx$ is denoted by
$\norm{\vecx}_p = \left(\sum_{i=1}^{d} |x_i|^p \right)^{1/p}$.
In particular, the $L_\infty$-norm of $\vecx$ is
$\norm{\vecx}_\infty = \max_i |x_i|$.
Moreover, for a norm $\norm{\cdot}$,
$\norm{\vecx}_*$ denotes the dual norm of $\norm{\vecx}$,
where
$\norm{\vecx}_* = \sup \{\langle \bm{x},\bm{z}\rangle
	\mid \norm{\vecz} \leq 1 \}$.
A norm $\norm{\cdot}$ over $\Real^d$ is monotone if
$\norm{\vecx} \leq \norm{\vecy}$ whenever
$|x_i| \leq |y_i|$ for every $1 \leq i \leq d$.
Note that $L_p$-norm is monotone for any $p \geq 1$.

\subsection{Online load balancing}

Firstly we begin with a standard (offline) load balancing
problem.
Suppose that we have $K$ servers to do
a simple task with a large amount of data.
The task can be easily parallelized in such a way that
we can break down the data into $K$ pieces and assign them to
the servers, and then each server processes the subtask in time
proportional to the size of data assigned.
An example is to find blacklisted IP addresses in
an access log data.
Each server is associated with loaded condition,
expressed in terms of ``the computation time per unit data''.
The goal is to find a data assignment to the servers
so as to equalize the computation time for all servers.
In other words, we want to minimize the \emph{makespan}, defined
as the maximum of the computation time over all servers.

Formally, the problem is described as follows:
The input is a $K$-dimensional vector
$\bm{l}=(l_1, l_2, \ldots, l_K) \in \Real_+^K$, where each $l_i$
represents the loaded condition of the $i$-th server.
The output is a $K$-dimensional probability vector
$\vecalpha = (\alpha_1, \alpha_2, \ldots, \alpha_K) \in
\Delta(K) = \{ \vecalpha \in [0,1]^K \mid \sum_{i=1}^K \alpha_i = 1 \}$,
where each $\alpha_i$ represents the fraction of data assigned to
the $i$-th server.
The goal is to minimize the makespan
$\norm{\vecalpha \odot \vecl}_\infty$, where
$\vecalpha \odot \vecl = (\alpha_1 l_1, \alpha_2 l_2, \ldots, \alpha_K l_K)$.
Note that it is clear that the optimal solution is given by
$\alpha_i = l_i^{-1}/\sum_{j=1}^K l_j^{-1}$, which equalizes
the computation time of every server as
\[
	C_\infty^*(\vecl) \stackrel{\text{def}}{=}
	\min_{\vecalpha \in \Delta(K)} \norm{\vecalpha \odot \vecl}_\infty
	= \frac{1}{\sum_{j=1}^K 1/l_j}.
\]
Note also that the objective is generalized to the $L_p$-norm
for any $p$ in the literature.

In this paper, we consider a more general objective
$\norm{\vecalpha \odot \vecl}$ for an arbitrary norm
that satisfies certain assumptions stated below.
In the general case, the optimal value is denoted by
\[
	C^*(\vecl) \stackrel{\text{def}}{=}
	\min_{\vecalpha \in \Delta(K)} \norm{\vecalpha \odot \vecl}.
\]
\begin{assumption}
Throughout the paper, we put the following assumptions on the norm.
\label{assumption:norm}
\begin{enumerate}
\item The norm $\norm{\cdot}$ is monotone, and
\item The function $C^*$ is concave.
\end{enumerate}
\end{assumption}
Note that the first assumption is natural for
load balancing and the both assumptions are satisfied by
$L_p$-norm for $p > 1$.

Now we proceed to the online load balancing problem with respect to
a norm $\norm{\cdot}$ that satisfies Assumption~\ref{assumption:norm}.
The problem is described as a repeated game between the learner
and the environment who may behave adversarially.
In each round $t = 1, 2, \ldots, T$, the learner
chooses an assignment vector
$\vecalpha_t \in \Delta(K)$, and then receives from the environment
a loaded condition vector $\vecl_t \in [0,1]^K$, which may
vary from round to round.
After the final round is over, the performance of the learner is
naturally measured by
$\norm{\sum_{t=1}^T \vecalpha_t \odot \vecl_t}$.
We want to make the learner perform nearly as well as
the performance of the best fixed assignment in hindsight
(offline optimal solution), which is given by $C^*(\sum_{t=1}^T \vecl_t)$.
To be more specific,
the goal is to minimize the following \emph{regret}:
\[
	\text{Regret}(T) = \norm{\sum_{t=1}^T \vecalpha_t \odot \vecl_t}
		- C^*\left(\sum_{t=1}^T \vecl_t\right).
\]

\subsection{Repeated game with vector payoffs and
approachability}

We briefly review the notion of Blackwell's approachability,
which is defined for a repeated game with vector payoffs.
The game is specified by a tuple
$(A, B, r, S, \text{dist})$, where
$A$ and $B$ are convex and compact sets,
$r:A \times B \to \Real^d$ is a vector-valued payoff function,
$S \subseteq \Real^d$ is a convex and closed set called
the \emph{target set}, and
$\text{dist}:\Real^d \times \Real^d \to \Real_+$ is a metric.
The protocol proceeds in trials:
In each round $t=1, 2, \ldots, T$,
the learner chooses a vector $\veca_t \in A$,
the environment chooses a vector $\vecb_t \in B$, and then
the learner obtains a vector payoff $\vecr_t \in \Real^d$, given by
$\vecr_t = r(\veca_t, \vecb_t)$.
The goal of the learner is to make the average payoff vector
arbitrarily close to the target set $S$.

\begin{definition}[Approachability]
For a game $(A,B,r,S,\mathrm{dist})$, the target set $S$ is
approachable with convergence rate $\gamma(T)$ if
there exists an algorithm for the learner such that
the average payoff $\bar \vecr_T = (1/T) \sum_{t=1}^T \vecr_t$
satisfies
\[
	\mathrm{dist}(\bar \vecr_T,S) \stackrel{\mathrm{def}}{=}
	\min_{\vecs \in S} \mathrm{dist}(\bar \vecr_T, \vecs) \leq \gamma(T)
\]
against any environment.
In particular, we simply say that $S$ is approachable if
it is approachable with convergence rate $o(T)$.
\end{definition}

Blackwell characterizes the approachability in terms of the
support function as stated in the proposition below.

\begin{definition}
For a set $S \subseteq \Real^d$,
the support function $h_S: \Real^d \to \Real \cup \{\infty\}$
is defined as
\[
	h_S(\vecw) = \sup_{\vecs \in S} \ip{\vecs}{\vecw}.
\]
\end{definition}
It is clear from definition that $h_S$ is convex whenever
$S$ is convex.

\begin{definition}[Blackwell~\cite{blackwell1956analog}]
A game $(A,B,r,S,\mathrm{dist})$ satisfies Blackwell Condition,
if and only if
\begin{equation}
\label{eq:BlackwellCondition}
	\forall \vecw \in \Real^d \;
	\left( \min_{\veca \in A} \min_{\vecb \in B}
		\ip{\vecw}{r(\veca, \vecb)} \leq h_S(\vecw)
	\right).
\end{equation}
\end{definition}

\begin{remark}
In \cite{blackwell1956analog}, Blackwell characterized the approachability of a target set for $L_{2}$-norm metric in terms of the Blackwell condition.
\end{remark}

In what follows, we only consider a norm metric, i.e,
$\mathrm{dist}(\vecr,\vecs) = \norm{\vecr - \vecs}$
for some norm $\norm{\cdot}$ over $\Real^d$.
The following proposition is useful.

%\begin{proposition}[\cite{shimkin2016online}]
%\label{lemma:distance}
%Let $S$ be a closed convex set with support function
%$h_S$ and let
%$\mathrm{dist}(\vecz,S)=\min_{\vecs \in S} \norm{\vecz - \vecs}$
%denote the point-to-set distance with respect to a norm $\norm{\cdot}$.
%Then, for any $\vecz \in \Real^d$,
%\[
%	\mathrm{dist}(\vecz,S) =
%		\max_{\vecw: \norm{\vecw}_* \leq 1}
%			\{ \ip{\vecw}{\vecz} - h_S(\vecw) \}.
%\]
%\end{proposition}

\begin{proposition} \label{prop:subgradient-support}
For any $\vecw \in \Real^d$,
$\vecs^* = \arg\max_{\vecs \in S} \ip{\vecs}{\vecw}$
is a sub-gradient of $h_S(\vecw)$ at $\vecw$.
\end{proposition}

\begin{proof}
For any $\vecw, \vecu \in \Real^d$, let
$\vecs^* = \argmax_{\vecs \in S} \ip{\vecs}{\vecw}$ and
$\vecs^\vecu = \argmax_{\vecs \in S} \ip{\vecs}{\vecu}$.
Since $\ip{\vecs^*}{\vecu} \leq \ip{\vecs^\vecu}{\vecu}$, we have
\begin{align*}
	h_S(\vecw) - h_S(\vecu) & =
	\sup_{\vecs \in S} \ip{\vecs}{\vecw}
		- \sup_{\vecs \in S}\ip{\vecs}{\vecu}
		= \ip{\vecs^*}{\vecw} - \ip{\vecs^\vecu}{\vecu} \\
	& \leq
	\ip{\vecs^*}{\vecw - \vecu},
\end{align*}
which implies the proposition.\qed

\end{proof}

\subsection{Online convex optimization}
\label{sec:OCO}

In this subsection we briefly review online convex optimization
with some known results. See, e.g., \cite{shalev2012online,hazan:book16} for more details.

An online convex optimization (OCO) problem is specified by
$(W, F)$, where $W \subseteq \Real^d$ is a compact convex set
called the decision set and
$F \subseteq \{f:W \to \Real \}$ is a set of convex functions over $W$
called the loss function set.
The OCO problem $(W,F)$ is described by the following protocol between
the learner and the adversarial environment.
For each round $t=1,2,\ldots,T$,
the learner chooses a decision vector $\vecw_t \in W$ and then
receives from the environment a loss function $f_t \in F$.
In this round, the learner incurs the loss given by $f_t(\vecw_t)$.
The goal is to make the cumulative loss of the learner nearly as
small as the cumulative loss of the best fixed decision.
To be more specific, The goal is to minimize the following regret:
\[
	\mathrm{Regret}_{(W,F)}(T) = \sum_{t=1}^T f_t(\vecw_t)
		- \min_{\vecw \in W} \sum_{t=1}^T f_t(\vecw).
\]
Here we add the subscript $(W,F)$ to distinguish from the regret
for online load balancing.

Any OCO problem can be reduced to
an online linear optimization (OLO) problem, which is an OCO
problem with linear loss functions.
More precisely, an OLO problem is
specified by $(W,G)$, where $G \subseteq \Real^d$ is the set
of cost vectors such that the loss function at round $t$ is
$\ip{\vecg_t}{\cdot}$ for some cost vector $\vecg_t \in G$.
For the OLO problem $(W,G)$, the regret of the learner is thus given by
\[
	\mathrm{Regret}_{(W,G)}(T) = \sum_{t=1}^T \ip{\vecg_t}{\vecw_t}
		- \min_{\vecw \in W} \sum_{t=1}^T \ip{\vecg_t}{\vecw}.
\]
The reduction from OCO to OLO is simple.
Run any algorithm for OLO $(W,G)$ with
$\vecg_t \in \partial f_t(\vecw_t)$, and then it achieves
$\mathrm{Regret}_{(W,F)}(T) \leq \mathrm{Regret}_{(W,G)}(T)$,
provided that $G$ is large enough, i.e.,
$G \supseteq \bigcup_{f \in F, \vecw \in W} \partial f(\vecw)$.

A standard FTRL (follow-the-regularized-leader) strategy
for the OLO problem $(W,G)$ is to choose $\vecw_t$ as
\begin{equation}
\label{eq:FTRL}
	\vecw_t = \arg\min_{\vecw \in W} \left(
		\sum_{s=1}^{t-1} \ip{\vecg_s}{\vecw} +
		\eta_t R(\vecw)
	\right),
\end{equation}
where $R:W \to \Real$ is a strongly convex function called the
regularizer and $\eta_t \in \Real_+$ is a parameter.
Using the strategy (\ref{eq:FTRL})
the following regret bound is known.

\begin{proposition}[\cite{shalev2012online}]
\label{prop:FTRL}
 Suppose that the regularizer $R:W \to \Real$ is $\sigma$-strongly convex
  w.r.t. some norm $\|\cdot\|$, i.e., for any $\vecw, \vecu \in W$, for any
  $\vecz \in \partial R(\vecw)$, $R(\vecu) \geq R(\vecw) + \langle
  \vecz,\vecu -\vecw\rangle + \frac{\sigma}{2}\|\vecu -\vecw\|^2
  $.
Then, for the OLO problem $(W,G)$, the regret of the strategy
(\ref{eq:FTRL}) satisfies
\[
	\mathrm{Regret}_{(W,G)}(T) = O(D_R L_G\sqrt{T/\sigma}),
\]
  where $D_R = \sqrt{\max_{\vecw \in W}R(\vecw)}$, $L_G=\max_{\vecg \in
  G}\|\vecg\|_*$ and $\eta_t=(L_G/D_R)\sqrt{T/\sigma}$.
\end{proposition}

Note however that the strategy does not consider the
computational feasibility at all.
For efficient reduction, we need an efficient algorithm
that computes a sub-gradient $\vecg \in \partial f(\vecw)$
when given (a representation of) $f \in F$ and $w \in W$,
and an efficient algorithm
for solving the convex optimization problem (\ref{eq:FTRL}).

For a particular OLO problem $(W,G)$ with $L_1$ ball decision set
$W = \{ \vecw \in \Real^d \mid \norm{\vecw}_1 \leq 1 \}$,
an algorithm called EG$^\pm$~\cite{hoeven2018many} finds
in linear time the optimal solution of (\ref{eq:FTRL}) with an entropic
regularizer and achieves the following regret.

\begin{theorem}[\cite{kivinen1997exponentiated}]
\label{theorem:EG-plus-minus}
For the OLO problem $(W,G)$ with
$W = \{ \vecw \in \Real^d \mid \norm{\vecw}_1 \leq 1 \}$
and $G = \{ \vecg \in \Real^d \mid \norm{\vecg}_\infty \leq M \}$,
EG$^\pm$ achieves
\[
	\mathrm{Regret}_{(W,G)}(T) \leq M\sqrt{2T \ln (2d)}.
\]
\end{theorem}

%\begin{theorem} \cite{shimkin2016online}
%For any convex and closed set $S,$ $S$ is approachable with respect to $\Vert \cdot \Vert$-norm if and only if for every $\Vert w \Vert_{*} \leq 1,$ where $\Vert \cdot \Vert_{*}$ denotes the dual norm of $\Vert \cdot \Vert,$ there exists $\alpha \in \Delta(K)$ so that
%\begin{equation}
%\langle w, r(\alpha, l)\rangle -h_{S}(w) \leq 0, \quad \forall l\in [0,1]^{K}.
%\end{equation}
%\end{theorem}

%\input{alg}
\section{Main result}

In this section, we propose a meta-algorithm for online load balancing,
which is obtained by combining a reduction to two independent OLO
problems and an OLO algorithm (as an oracle) for the reduced problems.
Note that the reduced OLO problems depend on the choice of norm
for online load balancing, and the OLO problems are further
reduced to some optimization problems defined in terms of the norm.
For efficient implementation, we assume that the optimization
problems are efficiently solved.

Now we consider the online load balancing problem on $K$ servers
with respect to a norm $\norm{\cdot}$ defined over $\Real^K$
that satisfies Assumption~\ref{assumption:norm}.
The reduction we show consists of three reductions,
the first reduction is to a repeated game with vector payoffs,
the second one is to an OCO problem, and
the last one is to two OLO problems.
In the subsequent subsections, we give these reductions, respectively.

\subsection{Reduction to a vector payoff game}

% \comm{The following lemma is irrelevant to the reduction.}
% Firstly of all, we give a Lemma for sub-gradient of $f_{t}(\bm{w}).$
% \begin{lemma}\label{lemma:sub-gradient}
% For our target set $S,$ and the loss function $f_{t}(\vecw)=\langle -\bm{r}_{t}, \vecw\rangle +h_{S}(\vecw),$ so we have that for any $\bm{z} \in \partial f_{t}(\vecw)$:
% $$\Vert \bm{z} \Vert _{p} \leq (2K)^{1/p}.$$
% \end{lemma}
% We give the proof of this Lemma in Appendix.

We will show that the online load balancing problem can be reduced to
the following repeated game with vector payoffs, denoted by
$P = (A,B,r,S,\mathrm{dist})$, where
\begin{itemize}
\item $A = \Delta(K)$, \quad  $B = [0,1]^K$,
\item $r: A \times B \to \Real^K \times \Real^K$ is the payoff function
defined as
$r(\vecalpha, \vecl) = (\vecalpha \odot \vecl, \vecl)$,
\item $S = \{ (\vecx,\vecy) \in [0,1]^K \times [0,1]^K \mid
	\norm{\vecx} \leq C^*(\vecy) \}$, and
\item $\mathrm{dist}$ is the metric over $\Real^K \times \Real^K$
defined as
$\mathrm{dist}(\vecr, \vecs) = \norm{\vecr - \vecs}^+$,
where $\norm{ \cdot }^+$ is the norm over $\Real^K \times \Real^K$
defined as
\[
	\norm{(\vecx,\vecy)}^+ = \norm{\vecx} + \norm{\vecy}.
\]
\end{itemize}
Here we use the convention that $\Real^{2K} = \Real^K \times \Real^K$.
Note that the target set $S$ is convex since $\norm{\cdot}$ is convex and
$C^*$ is concave by our assumption.
Note also that it is easy to verify that
$\norm{\cdot}^+$ is a norm whenever $\norm{\cdot}$ is a norm,
and its dual is
\begin{equation}
\label{eq:dual}
	\norm{(\vecx,\vecy)}^+_* = \max\{ \norm{\vecx}_*, \norm{\vecy}_* \}.
\end{equation}

The reduction is similar to that in \cite{even2009online},
but they consider a fixed norm $\norm{\cdot}_2$ to define
the metric, no matter what norm is used for online load balancing.

\begin{proposition}
\label{prop:OLBtoGame}
Assume that we have an algorithm for the repeated game $P$
that achieves convergence rate $\gamma(T)$.
Then, the algorithm, when directly applied to the online load
balancing problem, achieves
\[
	\mathrm{Regret}(T) \leq T \gamma(T).
\]
\end{proposition}

\begin{proof}
Let $\mathcal{A}$ denote an algorithm for the repeated game $P$
with convergence rate $\gamma(T)$.
Assume that when running $\mathcal{A}$ against the environment of
online load balancing, we observe, in each round $t$,
$\alpha_t \in \Delta(K)$ output from $\mathcal{A}$
and $\vecl_t \in [0,1]^K$ output from the environment.

Let $(\vecx,\vecy) = \arg\min_{(\vecx,\vecy) \in S}
\norm{\bar r_T - (\vecx,\vecy)}^+$, where
$\bar r_T = (1/T) \sum_{t=1}^T r(\vecalpha_t,\vecl_t)$
is the average payoff.
Note that by the assumption of $\mathcal{A}$, we have
$\norm{\bar r_T - (\vecx,\vecy)}^+ \leq \gamma(T)$.
For simplicity, let
\[
	L^\mathcal{A}_T = (1/T) \sum_{t=1}^T \vecalpha_t \odot \vecl_t
\text{\qquad and \qquad}
	L_T = (1/T) \sum_{t=1}^T \vecl_t.
\]
Then, we have
\begin{align*}
	(1/T) \mathrm{Regret}(T) &=
	\norm{L^\mathcal{A}_T} - C^*(L_T)
	= \bigl[ \norm{\vecx} - C^*(\vecy) \bigr] +
		\left[ \norm{L^\mathcal{A}_T} - \norm{\vecx} \right] +
		\bigl[ C^*(\vecy) - C^*(L_T) \bigr] \\
	& \leq
		\norm{L^\mathcal{A}_T - \vecx}
		+ \left[
			\min_{\vecalpha \in \Delta(K)} \norm{\vecalpha \odot \vecy}
			- \min_{\vecalpha \in \Delta(K)} \norm{\vecalpha \odot L_T}
		\right] \\
	& \leq
		\norm{L^\mathcal{A}_T - \vecx}
		+ \max_{\vecalpha \in \Delta(K)} \bigl[
			\norm{\vecalpha \odot \vecy} - \norm{\vecalpha \odot L_T}
		\bigr] \\
	& \leq
		\norm{L^\mathcal{A}_T - \vecx}
		+ \max_{\vecalpha \in \Delta(K)}
			\norm{\vecalpha \odot (\vecy - L_T)} \\
	& \leq
		\norm{L^\mathcal{A}_T - \vecx}
		+ \norm{\vecy - L_T} \\
	& = \norm{(L^\mathcal{A}_T,L_T) - (\vecx,\vecy)}^+ \\
	& = \norm{\bar r_T - (\vecx,\vecy)}^+ \\
	& \leq \gamma(T),
\end{align*}
where the first inequality is from the definition of $S$
and the triangle inequality,
the third inequality is from the triangle inequality,
and the fourth inequality is from the monotonicity of the
norm. \qed
\end{proof}

\subsection{Reduction to an OCO problem}

Next we give the second sub-reduction from the repeated game $P$
to an OCO problem.
We just follow a general reduction technique of
Shimkin~\cite{shimkin2016online} as
given in the next theorem.

\begin{theorem}[\cite{shimkin2016online}]
\label{theorem:shimkin-result}
Let $(A,B,r,S,\mathrm{dist})$ be a repeated game with vector payoffs,
where $\mathrm{dist}(\vecr,\vecs) = \norm{\vecr - \vecs}$
for some norm $\norm{ \cdot }$ over $\Real^d$.
Assume that we have an algorithm $\mathcal{A}$ that witnesses
the Blackwell condition,
i.e., when given $\vecw \in \Real^d$, $\mathcal{A}$ finds $\veca \in A$
such that $\ip{\vecw}{r(\veca,\vecb)}\leq h_S(\vecw)$ for
any $\vecb \in B$.
Assume further that we have an algorithm $\mathcal{B}$
for the OCO problem $(W,F)$,
where $W = \{ \vecw \in \Real^d \mid \norm{\vecw}_* \leq 1 \}$
and $F = \{ f:\vecw \mapsto \ip{-r(\veca,\vecb)}{\vecw} + h_S(\vecw)
\mid \veca \in A, \vecb \in B \}$.
Then, we can construct an algorithm for the repeated game
such that its convergence rate $\gamma(T)$ satisfies
\[
	\gamma(T) \leq \frac{\mathrm{Regret}_{(W,F)}(T)}{T}.
\]
Moreover, the algorithm runs in polynomial time (per round)
if $\mathcal{A}$ and $\mathcal{B}$ are polynomial time
algorithms.
\end{theorem}
For completeness, we give the reduction algorithm in Appendix.

The rest to show in this subsection is to ensure the existence
of algorithm $\mathcal{A}$
required for the reduction as stated in the theorem above.
In other words, we show that the Blackwell condition holds
for our game
$P = (\Delta(K), [0,1]^K, r, S, \mathrm{dist})$, where
$r(\vecalpha,\vecl) = (\vecalpha \odot, \vecl, \vecl) \in
\Real^K \times \Real^K$,
$S = \{(\vecx,\vecy) \in [0,1]^K \times [0,1]^K \mid
	\norm{\vecx} \leq C^*(\vecy) \}$, and
$\mathrm{dist}(\vecr,\vecs) = \norm{\vecr - \vecs}^+$.

\begin{lemma}\label{lemma:approachability}
The Blackwell condition holds for game $P$. That is,
for any $\vecw \in \Real^K \times \Real^K$,
we have
\[
	\min_{\vecalpha \in \Delta(K)} \max_{\vecl \in [0,1]^K}
		\ip{\vecw}{r(\vecalpha, \vecl)} \leq h_S(\vecw).
\]
\end{lemma}
Before we give the proof of Lemma, we need to involve a theorem as follow.
\begin{theorem}[\cite{cesa2006prediction}]
\label{theorem:neumann}
Let $f(x,y)$ denote a bounded real-valued function defined on
 $X \times Y$, where $X$ and $Y$ are convex sets and $X$ is compact.
Suppose that $f(\cdot, y)$ is convex and continuous for each fixed
$y \in Y$ and $f(x, \cdot)$ is concave for each fixed $x\in X$.
Then
\[
	\inf _{x \in X}\sup_{y \in Y}f(x,y)=\sup_{y \in Y}\inf_{x \in X}f(x,y).
\]
\end{theorem}

\begin{proof}[Proof of Lemma \ref{lemma:approachability}]
Let $\vecw = (\vecw_1,\vecw_2) \in \Real^K \times \Real^K$.
By the definition of $r$, the inner product in the Blackwell
condition can be rewritten as a bilinear function
\[
	f(\vecalpha ,\vecl)=\ip{\vecw}{r(\vecalpha, \vecl)}
	= \sum_{i=1}^K w_{1,i} \alpha_i l_i
		+ \sum_{i=1}^{K} w_{2,i} l_i
\]
over $\Delta(K) \times [0,1]^K$.
Therefore, $f$ meets the condition of Theorem~\ref{theorem:neumann}.
and we have
\[
	\min_{\vecalpha \in \Delta(K)} \max_{\vecl \in [0,1]^K}
		f(\vecalpha, \vecl) =
	\max_{\vecl \in [0,1]^K} \min_{\vecalpha \in \Delta(K)}
		f(\vecalpha, \vecl).
\]
Let
$\vecl^* = \arg\max_{\vecl \in [0,1]^K} \min_{\vecalpha \in \Delta(K)}
f(\vecalpha,\vecl)$
and $\vecalpha^* = \arg\min_{\vecalpha \in \Delta(K)}
\norm{\vecalpha \odot \vecl^*}$.
Note that by the definition of $S$, we have
$(\vecalpha^{*} \odot \vecl^*, \vecl^*) \in S$.
Hence we get
\begin{align*}
	\min_{\vecalpha \in \Delta(K)} \max_{\vecl \in [0,1]^K}
		f(\vecalpha, \vecl)
	&=
	\max_{\vecl \in [0,1]^K} \min_{\vecalpha \in \Delta(K)}
		f(\vecalpha, \vecl) \\
	&=
	f(\vecalpha^*, \vecl^*) \\
	&=
	\ip{\vecw}{((\vecalpha^{*} \odot \vecl^*),\vecl^*)} \\
	& \leq
	\sup_{\vecs \in S} \ip{\vecw}{\vecs} \\
	& =
	h_S(\vecw),
\end{align*}
which completes the lemma. \qed
\end{proof}

This lemma ensures the existence of algorithm $\mathcal{A}$.
On the other hand, for an algorithm $\mathcal{B}$
we need to consider the OCO problem $(W,F)$, where
the decision set is
\begin{equation}
\label{eq:DecisionSet}
	W = \{ \vecw \in \Real^K \times \Real^K
		\mid \norm{\vecw}^+_* \leq 1 \},
\end{equation}
and the loss function set is
\begin{equation}
\label{eq:CostSet}
	F = \{ f:\vecw \mapsto \ip{-r(\vecalpha,\vecl)}{\vecw} + h_S(\vecw)
		\mid \vecalpha \in \Delta(K), \vecl \in [0,1]^K \}.
\end{equation}
Since $W$ is a compact and convex set and
$F$ consists of convex functions, we could apply a number of
existing OCO algorithms to obtain
$\mathrm{Regret}_{(W,F)}(T) = O(\sqrt{T})$.
In the next subsection, we show that the problem can be simplified
to two OLO problems.

\subsection{Reduction to two OLO problems}

Consider the OCO problem $(W,F)$ given by
(\ref{eq:DecisionSet}) and (\ref{eq:CostSet}).
Following the standard reduction technique from OCO to OLO
stated in Section~\ref{sec:OCO},
we obtain an OLO problem $(W,G)$ to cope with, where
$G \subseteq \Real^K \times \Real^K$ is any set of cost vectors
that satisfies
\begin{equation}
\label{eq:CostVectorSet}
	G \supseteq \bigcup_{f \in F, \vecw \in W} \partial f(\vecw)
	= \left\{ -r(\vecalpha,\vecl) + \vecs \Bigm| \vecalpha \in \Delta(K),
		\vecl \in [0,1]^K,
		\vecs \in \textstyle\bigcup_{\vecw \in W} \partial h_S(\vecw)
	\right\}.
\end{equation}
By (\ref{eq:dual}), the decision set $W$ can be rewritten as
$W = B_*(K) \times B_*(K)$ where
$B_*(K) = \{ \vecw \in \Real^K \mid \norm{\vecw}_* \leq 1\}$ is the
$K$-dimensional unit ball with respect to the dual norm $\norm{ \cdot }_*$.
By Proposition~\ref{prop:subgradient-support},
any $\vecs \in \partial h_S(\vecw)$ is in the target set $S$,
which is a subset of $[0,1]^K \times [0,1]^K$.
Moreover, $r(\vecalpha,\vecl) = (\vecalpha \odot \vecl, \vecl)
\in [0,1]^K \times [0,1]^K$ for any $\vecalpha \in \Delta(K)$
and $\vecl \in [0,1]^K$.
Therefore,
$G = [-1,1]^K \times [-1,1]^K$ satisfies (\ref{eq:CostVectorSet}).

Thus, $(B_*(K) \times B_*(K), [-1,1]^K \times [-1,1]^K)$
is a suitable OLO problem reduced from the OCO problem $(W,F)$.
Furthermore, we can break the OLO problem
into two independent OLO problems $(B_*(K), [-1,1]^K)$
in the straightforward way:
Make two copies of an OLO algorithm $\mathcal{C}$ for
$(B_*(K), [-1,1]^K)$, denoted by $\mathcal{C}_1$ and $\mathcal{C}_2$,
and use them for predicting the first half and second half
decision vectors, respectively. More precisely, for each trial $t$,
(1)
receive predictions $\vecw_{t,1} \in B_*(K)$ and
$\vecw_{t,2} \in B_*(K)$ from $\mathcal{C}_1$ and $\mathcal{C}_2$,
respectively,
(2) output their concatenation $\vecw_t = (\vecw_{t,1}, \vecw_{t,2})
\in W$,
(3) receive a cost vector
$\vecg_t = (\vecg_{t,1}, \vecg_{t,2}) \in [0,1]^K \times [0,1]^K$ from
the environment,
(4) feed $\vecg_{t,1}$ and $\vecg_{t,2}$
to $\mathcal{C}_1$ and $\mathcal{C}_2$, respectively, to
make them proceed.

It is clear that the procedure above ensures the
following lemma.

\begin{lemma} \label{lemma:OCOtoOLO}
The OCO problem $(W,F)$ defined as
(\ref{eq:DecisionSet}) and (\ref{eq:CostSet}) can be
reduced to the OLO problem $(B_*(K),[0,1]^K)$, and
\[
	\mathrm{Regret}_{(W,F)}(T) \leq 2 \mathrm{Regret}_{(B_*(K),[0,1]^K)}(T).
\]
\end{lemma}

\subsection{Putting all the pieces together}

Combining all reductions stated in the previous subsections,
we get an all-in-one algorithm as described in
Algorithm \ref{alg:all-in-one}.

\begin{algorithm}
\caption{An OLO-based online load balancing algorithm}
\label{alg:all-in-one}
\begin{algorithmic}
\REQUIRE An algorithm $\mathcal{A}$ that, when given $\vecw$,
	finds $\vecalpha = \arg\min\limits_{\vecalpha \in \Delta(K)}
	\max\limits_{\vecl \in [0,1]^K}
	\ip{\vecw}{(\vecalpha \odot \vecl, \vecl)}$.
\REQUIRE An algorithm $\mathcal{B}$ that, when given $\vecw$,
	finds $\vecs \in \partial h_S(\vecw)$.
\REQUIRE Two copies of an algorithm, $\mathcal{C}_1$ and $\mathcal{C}_2$,
	for the OLO problem $(B_*(K),[-1,1]^K)$.
\FOR {$t=1, 2, \ldots, T$}
\STATE 1. Obtain $\vecw_{t,1}$ and $\vecw_{t,2}$ from
	$\mathcal{C}_1$ and $\mathcal{C}_2$, respectively, and
	let $\vecw_t = (\vecw_{t,1},\vecw_{t,2})$.
\STATE 2. Run $\mathcal{A}(\vecw_t)$ and obtain $\vecalpha_t \in \Delta(K)$.
\STATE 3. Output $\vecalpha_t$ and observe $\vecl_t \in [0,1]^K$.
\STATE 4. Run $\mathcal{B}(\vecw_t)$ and obtain
	$\vecs_t = (\vecs_{t,1},\vecs_{t,2})$.
\STATE 5. Let $\vecg_{t,1} = -\vecalpha_t \odot \vecl_t + \vecs_{t,1}$
	and $\vecg_{t,2} = -\vecl_t + \vecs_{t,2}$.
\STATE 6. Feed $\vecg_{t,1}$ and $\vecg_{t,2}$ to
	$\mathcal{C}_1$ and $\mathcal{C}_2$, respectively.
\ENDFOR
\end{algorithmic}
\end{algorithm}

It is clear that combining
Proposition~\ref{prop:OLBtoGame},
Theorem~\ref{theorem:shimkin-result}
and Lemma~\ref{lemma:OCOtoOLO},
we get the following regret bound of
Algorithm~\ref{alg:all-in-one}.

\begin{theorem} \label{regret:global-cost}
Algorithm~\ref{alg:all-in-one} achieves
\[
	\mathrm{Regret}(T) \leq 2 \mathrm{Regret}_{(B_*(K),[-1,1]^K)}(T),
\]
where the regret in the right hand side is the regret of
algorithm $\mathcal{C}_1$ (and $\mathcal{C}_2$ as well).
Moreover, if $\mathcal{A}$, $\mathcal{B}$ and $\mathcal{C}_1$ runs
in polynomial time, then Algorithm~\ref{alg:all-in-one} runs
in polynomial time (per round).
\end{theorem}

When applying the FTRL strategy as in (\ref{eq:FTRL})
to the OLO problem $(B_*(K),[-1,1]^K)$
with a strongly convex regularizer $R$, Proposition~\ref{prop:FTRL}
implies the following regret bound.
\begin{corollary}
Assume that there exists a regularizer $R:B_*(K) \to \Real$
that is $\sigma$-strongly convex w.r.t.\ $L_1$-norm.
Then, there exists an algorithm for the online load balancing
problem that achieves
\[
	\mathrm{Regret}(T) = O(D_R \sqrt{T/\sigma}),
\]
where $D_R = \sqrt{\max_{\vecw \in B_*(K)}R(\vecw)}$.
\end{corollary}

In particular,
for the OLO problem $(B_1(K),[-1,1]^K)$, algorithm EG$^\pm$
achieves $\sqrt{2T \ln 4K}$ regret bound
as shown in Theorem~\ref{theorem:EG-plus-minus}.
Thus we have $O(\sqrt{T \ln K})$ regret bound
for online load balancing with respect to
$L_\infty$-norm (i.e., w.r.t.\ makespan),
which improves the bound of~\cite{even2009online} by a factor of
$\sqrt{\ln K}$.
Moreover, for $L_\infty$-norm, it turns out that
we have polynomial time algorithms for $\mathcal{A}$
and $\mathcal{B}$, which we will give in the next section.
We thus obtain the following corollary.

\begin{corollary}
There exists a polynomial time (per round) algorithm
for the online load balancing problem
with respect to $L_\infty$-norm that achieves
\[
	\mathrm{Regret}(T) \leq 2\sqrt{2T \ln 4K}.
\]
\end{corollary}

\newcommand{\vecbeta}{\bm{\beta}}
\def\commH#1{\textbf{KH:#1}}

 \section{Algorithmic details for $L_\infty$-norm}
In this section we give details of Algorithm
\ref{alg:all-in-one} for the makespan problem, i.e., for
$L_\infty$-norm.
%we give a
%polynomial time algorithm.

  \subsection{Computing $\alpha_t$}
%There are two procedures in
%Algorithm~\ref{alg:alg-for-global-cost-function}.
First, we give details of implementation of $\mathcal{A}$ in Algorithm~\ref{alg:all-in-one}.
Specifically, on the round $t,$ we need to choose $\vecalpha_{t}$, which is the optimal solution
of the problem in Lemma \ref{lemma:approachability}.  That is,
%firstly and we need to give the optimal vector $s\in S$ for the support function $h_{S}(w_{t}).$

%To calculate $\alpha_{t}$ we will utilize Linear Programming. Let us have a review to our problem. We need to find an $\alpha_{t}$ such that
%$$\max_{l\in [0,1]^{K}} \langle w_{t}, (\alpha \odot l,l) \rangle \leq h_{S}(w_{t}).$$
%Since the we have known that $S$ is approachable, the existence of $\alpha$ is undoubtable.
%So it is equivalent for us to solve the following optimization problem:
\begin{equation}
\label{problem:allocation-for-each-load}
\min_{\bm{\alpha} \in \Delta(K)}\max_{\vecl \in [0,1]^{K}}\langle \vecw_{1}, (\bm{\alpha} \odot \vecl)\rangle+\langle \vecw_{2},\vecl \rangle,
\end{equation}
where we set that $\vecw=(\vecw_{1},\vecw_{2})$ and
$\vecw_{1}$ and $\vecw_2$ are $K$-dimensional vectors,
respectively.
We see that the optimization of this objective function is defined by
$l_{i}=0$ if $w_{1,i} \cdot \alpha_i+w_{2,i} \leq 0,$ otherwise we let
$l_i=1.$
Hence we can convert our problem to choose $\vecalpha$ as
$$\min_{\bm{\alpha} \in \Delta(K)}\max_{\bm{l} \in [0,1]^{K}}\langle
\vecw_{1}, (\bm{\alpha} \odot \bm{l})\rangle+\langle \vecw_{2},\bm{l}
\rangle=\min_{\bm{\alpha} \in \Delta(K)}\sum_{i=1}^{K}\max \left\{0,
\alpha_i w_{1,i}+w_{2,i}\right\},$$
which is equivalent to
\begin{equation*}
\begin{split}
&\min_{\vecalpha \in \Delta(K),\vecbeta \geq \bm{0}} \sum_{i=1}^{K} \beta_i\\
&\mathrm{s.t.}~\beta_i \geq w_{1,i} \alpha_i+ w_{2,i} \quad \forall i=1,\dots,K.
%&\sum_{i=1}^{K} \alpha_i=1; \quad \beta_i,\alpha_i \geq 0 \quad
\end{split}
\end{equation*}
The above problem is a linear program with $O(K)$ variables and  $O(K)$
linear constraints. Thus, computing $\vecalpha_t$ in the problem (\ref{problem:allocation-for-each-load}) can
be solved in polynomial time.

\subsection{Computing subgradients $g_t$ for the $\infty$-norm}
The second component of Algorithm~\ref{alg:all-in-one}
is the algorithm $\mathcal{B}$, which computes subgradients $\vecs_t \in
\partial h_S(\vecw_t)$.
By Proposition \ref{prop:subgradient-support},
%$\vecg_t$ is given as
%$\vecg_t = -(\bm{\alpha}_t \odot \bm{l}_t, \bm{l}_t) + \bm{s}_t$ and
we have $\vecs_t=\arg\max_{\vecs \in S}\langle \bm{s}, \vecw_{t} \rangle.$
Recall that $S=\{ (\vecx, \vecy) \in [0,1]^K \times [0,1]^K \mid
\|\vecx\|_\infty \leq C_\infty^*(\vecy)\}$.
%
%
%First of all, $C^*_\infty(\vecy)$ is given as follows:
%\begin{proposition}[Even-Dar et al.\cite{even2009online}]
 %For $\vecy \in [0,1]^K$,
 %\[
 %C_\infty^*(\vecy)
 %=
 %\frac{1}{\sum_{i=1}^K \frac{1}{y_i}}.
 %\]
%\end{proposition}
%
%We give a more general result to describe the set $S,$ which is the constraint of above optimal problem.
%
%\begin{lemma}\label{lemma:calclulate-sub-gradient}
%Given $S$ as:
%\begin{equation*}
%S=\lbrace (\bm{x},\bm{y}) \in \mathbb{R}^{k} \times \mathbb{R}^{k}: x_{i},y_{i} \in [0,1]; C_{p}(\bm{x}) \leq C_{p}^{\star}(\bm{y}) \rbrace, \%quad \forall p \geq 2.
%\end{equation*}
%
%We have that
%\begin{equation*}
%\Vert \bm{x} \Vert_{p} \leq \min_{\bm{\alpha} \in \Delta(K)} \Vert \bm{\alpha} \odot \bm{y}\Vert_{p}, \forall (\bm{x},\bm{y}) \in S \Longleftrightarrow \Vert \bm{x} \Vert_{p} \leq \frac{ \prod_{i=1}^{K}y_{i}  }{ \left( \sum_{i=1}^{K}\prod_{i\neq j}y_{j}^{\frac{p}{p-1}}  \right)^{\frac{p-1}{p}}  }, \forall x_{i},y_{i} \in [0,1].
%\end{equation*}
%\end{lemma}
%We give the proof of this Lemma in Appendix.
In particular, the condition that $\|\vecx\|_\infty \leq C_\infty^*(\bm{y})$ can
be represented as
\begin{align*}
  &\max_i x_i \leq \min_{\vecalpha \in \Delta(K)}\Vert \bm{\alpha} \odot \vecy\Vert_{\infty}
 \iff
 x_i \leq \frac{1}{\sum_{j=1}^K \frac{1}{y_j}}, \text{$\forall i$}.
\end{align*}
Therefore, the computation of the subgradient $\vecs_t$ is formulated as
\begin{equation}\label{problem:original-optimization}
\begin{split}
&\max_{\vecx,\bm{y}\in [0,1]^K} \langle \vecw_{1}, \vecx \rangle+
 \langle \vecw_{2}, \bm{y} \rangle\\
 &\mathrm{s.t.}~x_{i} \leq \frac{1}{\sum_{j}\frac{1}{y_{j}}}, \quad \forall i=1,\dots,K.
 %\quad \forall i=1,\dots,K; \quad x_{i}, y_{i} \in [0,1] \quad \forall i.\\
\end{split}
\end{equation}
%\commH{This problem should be minimization problem?}

Now we show that there exists an equivalent second order cone
programming(SOCP) formulation (e.g., \cite{lobo1998applications}) for this problem.

First we give the definition of the second order cone programming, and
then we give a proposition, which states that our optimization problem
is equivalent to the second order cone programming.
%As we mentioned before now we are going to calculate the optimal vector $s_{t}^{*}$ with Second Order Conic Programming \cite{lobo1998applications}.

%%A conic optimization problem is one of the form
%%$$\min_{x} c^{T}x : Ax=b, x \in \mathcal{K},$$
%%where $\mathcal{K}$ is a given convex cone, that is a direct product of one of the three following types:

%%The non-negative orthant, $\mathbb{R}^{n}_{+}.$

%%The second order cone, $ \mathcal{Q}^{n}=\lbrace (x,t)\in \mathbb{R}^{n}_{+}: t \geq \Vert x \Vert_{2} \rbrace.$

%%The semi-defined cone, $ \mathcal{S}_{+}^{n}= \lbrace X=X^{T} \succeq  0\rbrace.$

%%These three types are called as tractable conic optimizations. Hence we involve the second-order cone optimization.
\begin{definition}
%%A problem is a second-order cone optimization problem if it is a tractable conic optimization, where the cone $\mathcal{K}$ is a product of second-order cones and possibly the non-negative orthant $\mathbb{R}_{+}^{n}.$
The standard form for the second order conic programming(SOCP) model is as follows:
$$\min_{\vecx} \langle \bm{c},\vecx \rangle\text{~s.t.}~A\vecx=\bm{b}, \Vert C_{i}\vecx+\bm{d}_{i} \Vert_{2} \leq \bm{e}_{i}^{\top} \bm{x}+\bm{f}_{i} \quad \text{for~}i=1, \cdots, m,$$
where the problem parameters are $\bm{c}\in \mathbb{R}^{n},$ $C_{i} \in \mathbb{R}^{n_{i} \times n},$ $\bm{d}_{i} \in \mathbb{R}^{n_{i}},$ $\bm{e}\in\mathbb{R}^{n},$ $\bm{f}_{i}\in \mathbb{R},$ $A\in\mathbb{R}^{p \times n},$ and $\bm{b}\in\mathbb{R}^{p}.$ $\vecx\in\mathbb{R}^{n}$ is the optimization variable.
\end{definition}

Then we obtain the following proposition.
\begin{proposition}\label{proposition:equivalence-bewteen-two-formations}
$\sum_{i=1}^{K}\frac{x^{2}}{y_{i}} \leq x,$ $x \geq 0$ and $y_{i} \geq 0$ is equivalent to $x^{2} \leq y_{i}z_{i},$ where $y_{i},z_{i} \geq 0$ and $\sum_{i=1}^{K}z_{i}=x.$
\end{proposition}
\begin{proof}
On the direction $``\Rightarrow"$

From $\sum_{i=1}^{k}\frac{x^{2}}{y_{i}} \leq x$ we obtain that $\sum_{i=1}^{k}\frac{x}{y_{i}} \leq 1.$ By setting $$z_{i}=x \cdot \frac{\frac{1}{y_{i}}}{\sum_{i}\frac{1}{y_{i}}},$$
we can have that $x^{2} \leq y_{i}z_{i},$ and $\sum_{i=1}^{k}z_{i}=x.$

On the other direction $``\Leftarrow"$
Due to $x^{2} \leq y_{i}z_{i},$ we have $\frac{x^{2}}{y_{i}} \leq z_{i}.$ So we have that
$$\sum_{i=1}^{k}\frac{x^{2}}{y_{i}} \leq \sum_{i=1}^{k}z_{i} =x.$$
\qed
\end{proof}

Again in our case we need find to the optimal vector $\vecs \in S,$ which
satisfies that $\bm{s}_{t}=\argmax_{\vecs \in S}\langle
\vecw_{t},\vecs\rangle.$
%For $L_\infty$-norm, $C^*(\vecell)$ is concave, thus $S$ be
%Note that by Lemma \ref{lemma:s-is-convex-set} $S$ is a cone and it is
%a convex set. Therefore
Then we can reduce our problem in following theorem.

\begin{theorem}\label{theorem:Problem-to-SOCP}
The optimization problem (\ref{problem:original-optimization}) can be
 solved by the second order cone programming.
\end{theorem}

\begin{proof}
To prove this theorem we only need to represent the original problem
 (\ref{problem:original-optimization}) as a standard form of the SOCP problem.
Note that we only consider the case that $y_{i}\neq 0$ for all
 $i=1,\dots,K$.  The case where $y_i=0$ for some $i$ is trivial.
 To see this, by definition of $S$, we know that for all $i,$ $x_{i}=0.$
Then, the resulting problem is a linear program, which is a special case
 of the SOCP.
 % If it exists $y_{i}=0,$ by the definition of $S$
Now we assume that $y_i \neq 0$ for $i=1,\dots, K$.
 For $x_{i} \leq \frac{1}{\sum_{j}\frac{1}{y_{j}}},$ we multiply $x_{i}$ on both sides and rearrange the inequality:
$$\sum_{j=1}^{K} \frac{x_{i}^{2}}{y_{j}} \leq x_{i}.$$
By Proposition~\ref{proposition:equivalence-bewteen-two-formations},
 this is equivalent with
% which implies the conic representation as
$$y_{j}z_{i,j} \geq x_{i}^{2}, \quad y_{j},z_{i,j} \geq 0, \quad \sum_{j=1}^{K}z_{i,j} =x_{i}.$$

By \cite{lobo1998applications}, we may rewrite it as follows: For each $i,$
\begin{equation}\label{to-SOCP}
x_{i}^{2} \leq y_{j}z_{i,j};  \quad y_{j}, z_{i,j} \geq 0 \Longleftrightarrow \Vert (2x_{i}, y_{j}-z_{i,j})\Vert_{2} \leq y_{j}+z_{i,j} \quad \forall j=1,\dots,K.
\end{equation}
The above equivalence is trivial.
On the other hand, since $x_{i} \leq \frac{1}{\sum_{j}\frac{1}{y_{j}}},$ and $y_{i} \in [0,1],$ naturally we have $x_{i} \in [0,1].$ So we need only constrain that $y_{i}\in [0,1].$ We can apply the face that if $y_{i}$ is positive so $\vert y_{i} \vert=y_{i},$ and if $y_{i} \leq 1,$ so $\vert y_{i} \vert \leq 1.$
Therefore we may give a $(K^{2}+2K) \times (K^{2}+2K)$-matrix $C_{i}$ in SOCP, and the variable vector is composed as follows:
\begin{equation}
\bm{\tilde{x}}=(x_{1}, \cdots, x_{K},y_{1}, \cdots, y_{K}, z_{1,1}, \cdots, z_{1,K}, \cdots, z_{K,1} \cdots, z_{K,K}),
\end{equation}
where for $z_{i,j},$ $i$ is corresponding to $x_{i}.$

Now we may give the second order cone programming of our target problem as follows:
\begin{equation}
\begin{split}
&\min_{\bm{\tilde{x}}} \langle -(\vecw_{1},\vecw_{2}, 0, \cdots,0), \bm{\tilde{x}} \rangle \\
&\mathrm{s.t.}\Vert C_{i} \bm{\tilde{x}} \Vert_{2} \leq \bm{e}_{i}^{\top} \bm{\tilde{x}} +\bm{d}_{i} \quad \forall i=1, \cdots, K^{2}+2K, \\
&A\bm{\tilde{x}}=\bm{b}.
\end{split}
\end{equation}
where $C_{i},$ $\bm{e}_{i}$, $A$ and $\bm{b}$ are defined as follows:

Firstly the matrix $C$ for hyperbolic constraints are given as:
For a fixed $s\in [K],$ where $[K]=\{1, \cdots, K\}$ in matrix $C_{i},$ where $i\in [(s-1)K, sK]$ we let $(C_{i})_{1,s}=2,$ $(C_{i})_{K+i,K+i}=1$, $(C_{i})_{2K+(s-1)K+i,2K+(s-1)K+i}=-1,$ and others are $0.$
$\bm{e}_{i}$ is defined as $(\bm{e}_{i})_{K+i}=1$ and $(\bm{e}_{i})_{2K+(s-1)K+i}=1,$ others are $0.$

Next we need to constrain that $y_{i}$ is less than $1$.
For $i\in [K^{2},K^{2}+K]$ we let that $(C_{i})_{K+i,K+i}=1$ and others are $0.$ And we let that $\bm{e}_{i}$ is a zero vector and $\bm{d}_{i}=1.$ It means that $\Vert y_{i}\Vert \leq 1.$ For $i \in [K^{2}+K, K^{2}+2K],$ we set $(C_{i})_{K+i,K+i}=1$ $\bm{e}_{K+i}=1,$ and $\bm{d}_{i}=0$

At last we need to constrain that $\sum_{j=1}^{K}z_{j}=x_{i}$ in equation \ref{to-SOCP}:
Let $A \in \mathbb{R}^{K \times (3K+K^{2})}$ for each row vector
 $A_{j},$ where $j \in [K],$ we have that $(A_{j})_{j}=1$ and
 $(A_{j})_{2K+(j-1)j+m}=-1,$ for all $m=1,\cdots, K.$ No w the matrix $A$ is composed by the row vectors $A_{j}.$ and $\bm{b}$ is a zero vector. \qed
\end{proof}

\section{Conclusion}
In this paper we give a framework for online load balancing problem by
reducing it to two OLO problems. Moreover, for online load balancing
problem with respect to $L_\infty$-norm we achieve the best known regret
bound in polynomial time.
Firstly, we reduce online load balancing with
$\Vert  \cdot \Vert$ norm to a vector payoff game measured by
combination norm $\Vert \cdot \Vert^{+}.$ Next due to
\cite{shimkin2016online} this vector payoff game is reduced to an OCO
problem.
At last, we can reduce this OCO problem to two independent OLO problems. Especially, for makespan, we give an efficient algorithm, which achieves the best known regret bound $O(\sqrt{T \ln K}),$ by processing linear programming and second order cone programming in each trial.

There are some open problems left in this topic. For instance, an
efficient algorithm for online load balancing with respect to general
norm or $p$-norm is still an open problem. Furthermore, the lower bound
of online load balancing is still unknown. 
\bibliographystyle{splncs04}
\bibliography{ref,ref_hatano}

\section{Appendix}
\subsection{A general reduction algorithm from
a repeated game to an OCO problem}

For completeness, we give
in Algorithm~\ref{alg:OCO-based-meta-algorithm}
a general reduction algorithm of
Shimkin~\cite{shimkin2016online} from
a repeated game with vector payoffs to an OCO problem.

\begin{algorithm}
\caption{Reduction from game
$(A,B,r,S,\mathrm{dist})$ with $\mathrm{dist}(\veca,\vecb)=
\norm{\veca-\vecb}$ to OCO~\cite{shimkin2016online}}
\label{alg:OCO-based-meta-algorithm}
\begin{algorithmic}
\REQUIRE An algorithm $\mathcal{A}$ that, when given $\vecw$,
finds $\veca \in A$ such that $\ip{\vecw}{r(\veca,\vecb)} \leq h_S(\vecw)$
for any $\vecb \in B$.
\REQUIRE An algorithm $\mathcal{B}$ for the OCO problem $(W,F)$,
where $W = \{ \vecw \mid \norm{\vecw}_* \leq 1 \}$ and
$F = \{ f:\vecw \mapsto \ip{-r(\veca,\vecb)}{\vecw} + h_S(\vecw)
	\mid \veca \in A, \vecb \in B \}$.
\FOR {$t=1, 2, \ldots, T$}
\STATE 1. Obtain $\vecw_t \in W$ from $\mathcal{B}$.
\STATE 2. Run $\mathcal{A}(\vecw_t)$ and obtain $\veca_t \in A$.
\STATE 3. Output $\veca_t \in A$ and observe $\vecb_t \in B$.
\STATE 4. Construct the loss function
		$f_t: \vecw \mapsto \ip{-r(\veca_t,\vecb_t)}{\vecw} + h_S(\vecw)$
		and feed it to $\mathcal{B}$.
\ENDFOR
\end{algorithmic}
\end{algorithm}

\end{document}